\documentclass[runningheads]{llncs}

\usepackage[T1]{fontenc}
\usepackage{graphicx}

\usepackage{hyperref}
\usepackage{color}

\usepackage{amsmath}
\usepackage{amsfonts}
\usepackage{amssymb}

\usepackage{booktabs}
\usepackage{comment}
\usepackage{algorithm}
\usepackage{algorithmic}

\newcommand{\Z}{\mathbb{Z}}
\newcommand{\LM}{{\rm LM}}
\newcommand{\LC}{{\rm LC}}
\newcommand{\LT}{{\rm LT}}
\newcommand{\lcm}{{\rm lcm}}
\newcommand{\In}{{\rm in}}
\newcommand{\es}{{\rm S}}
\newcommand{\el}{{\rm L}}
\newcommand{\sugar}{{\rm sugar}}
\newcommand{\E}{{\rm E}}
\newcommand{\q}[1]{\texttt{#1}}

\begin{document}

\title{Letting Homogeneity Entropy Select S-Pairs in Buchberger’s Algorithm}

\author{Uzma Shafiq\inst{1,2}       \orcidID{0000-0001-9146-2478} \and
Matthew England\inst{1}             \orcidID{0000-0003-1701-8090} \and
AmirHosein Sadeghimanesh\inst{1}   \orcidID{0000-0002-6945-3118} \and \\
Nayyar Zaidi\inst{2}                \orcidID{0000-0003-4024-2517} 
}
\authorrunning{U. Shafiq \textit{et al.}}

\institute{Coventry University, Coventry, UK \and
Deakin University, Melbourne, Australia
}

\maketitle              

\begin{abstract}
We present a novel S-pair selection strategy called Homogeneity Entropy, for deciding the sequence of S-polynomials to construct in Buchberger’s algorithm to compute a Gr\"{o}bner basis. The strategy uses an information-theoretic measure derived from the distribution of degrees among the monomials of the S-polynomial: a very different approach to the classical heuristics such as Degree, Normal and Sugar, or indeed the more recent machine learning approaches to the problem. 

We implement this strategy and evaluate it on two different datasets: (1) variations of randomly generated polynomial systems with controlled numbers of variables, degrees, densities and number of polynomials per system; and (2) the PHCpack benchmark dataset sourced from real world problems. The Homogeneity Entropy strategy significantly outperforms classical strategies on random polynomial datasets, but on the PHCpack dataset the classical strategies perform better.  This suggests the right strategy varies with the shape of the data and we explore this in several experiments.  The new strategy offers practically meaningful gains on certain distributions, and represents the first use of such information-theoretic guidance in the optimisation of symbolic computation algorithms.

\end{abstract}

\section{Introduction}

Gr\"{o}bner bases are a fundamental tool in computational algebra, widely used in applications such as algebraic geometry \cite{DK_86}, chemical reaction networks theory \cite{Sadeghimanesh_Feliu_2019}, coding theory \cite{CMF_08}, cryptography \cite{CMF_08,SMM_09} and systems theory \cite{BB_01}. There are multiple algorithms for computing Gr\"{o}bner bases, with the original, Buchberger’s algorithm \cite{Buch_06}, remaining conceptually central, and the F5 algorithm \cite{FJ_2002} generally considered to be the fastest. Constructing Gr\"{o}bner bases is known to require doubly exponential complexity in the worst-case scenario \cite{MEMA_1982}, but it has long been observed that significant optimisations are available.

Since \cite{HEWDPB14}, there has been an increasing amount of work exploring the use of machine learning tools to optimise computer algebra algorithms,  and covering topics such as the variable ordering for Cylindrical Algebraic Decomposition (e.g. \cite{dRE24,JDLHMZ23,PdREC24}), and  algorithm selection for symbolic integration (e.g. \cite{BEG24b,BSEG24,SNB23}). 

The first such paper for optimising Gr\"{o}bner basis construction was \cite{PSH_20} which used reinforcement learning to train an agent to decide S-pair selection in Buchberger's algorithm. Since then we have seen other machine learning optimisations in ideal theory \cite{HMS25} and even the direct machine learning generation of Gr\"{o}bner Bases via transformers \cite{KHI_23}.  

In this paper we are focused on the S-pair selection problem for Gr\"{o}bner Bases in Buchberger's algorithm.  Our paper stems from our recent recreation of the S-pair selection experiments in \cite{PSH_20}, where we extended the original methodology to consider alternative features including some from information theory. The strong predictive performance of features making use of Shannon entropy, which quantifies the degree of concentration (or uncertainty) of a distribution, led us to analyse this independently of the machine learning in the present paper.

\subsection{S-pair Selection}

Recall that in Buchberger's algorithm, the main loop processes a set of S-pairs one at a time, terminating when the set is empty: the act of processing one S-pair can potentially add new pairs to the set for consideration \cite{Buch_06}, but termination is guaranteed eventually. While the correctness of Buchberger’s algorithm is independent of the strategy on which S-pair to choose first, its efficiency is highly dependent on it, even when using reduction rules such as the Gebauer–M\"{o}ller criteria \cite{Gebauer_Moller_1988} which can avoid considering some pairs with redundant information.

Classical S-pair selection strategies such as Sugar selection \cite{Giovini_et_al_1991} try to control S-polynomial growth by prioritising pairs with low total degree, or minimal least common multiples of their leading monomials. While effective in many settings, such static heuristics optimise local properties of S-pairs and ignore the global structure of reductions that can occur during the computation \cite{MC_2004}. As a result, these strategies can perform suboptimally for ideals that exhibit strong structural dependencies or uneven reductions.  

\subsection{Contributions}

In this paper we introduce a new strategy: we interpret the normalised frequencies of term degrees in an S-polynomial as a probability distribution, and use Shannon entropy to measure the dispersion of polynomial’s monomials across different degrees, distinguishing degree-concentrated S-polynomials from degree-dispersed ones.  
We address the following research questions: 
\begin{itemize}
    \item RQ1. Can information-theoretic measures of S-polynomial structure serve as a useful signal for pair selection in Buchberger's algorithm?
    \item RQ2. How does the performance of an entropy-based selection strategy compare to classical strategies (Sugar, Degree)?
    \begin{itemize} 
        \item Across different distributions of randomly generated datasets.
        \item On the PHCpack benchmark systems dataset \cite{VPHCpack}.
    \end{itemize}
    \item RQ3. What input properties determine when our new selection strategy is competitive and when it is not?
\end{itemize}

    \section{Background and Preliminaries}
\label{sec:algebra}

The set of integers, and non-negative integers are denoted by $\Z$, and $\Z_{\geq 0}$ respectively. For a finite set $A$, its cardinality is denoted by $|A|$. When there is an implicit order among elements of $A$ and if $i$ is a positive integer less than or equal to $|A|$, then by $A_i$ we mean the $i$-th element of $A$.

Throughout this paper let $k$ be a field and $R=k[x_1,\dots,x_n]$ be a polynomial ring in the variables $x_1, \dots, x_n$ with coefficients from $k$. We assume an implicit order between the variables, and therefore a monomial $x_1^{a_1}x_2^{a_2}\cdots x_n^{a_n}$ can be written uniquely as $x^\alpha$ where $x=(x_1,\dots,x_n)$ and $\alpha=(a_1,\dots,a_n)$ without causing any confusion. A polynomial $f\in R$ can be written as $\sum_{\alpha\in A}c_\alpha x^\alpha$ where $A\subset\Z_{\geq 0}^n$ is a finite set. The degree of a term $c_\alpha x^\alpha$ (with $c_\alpha \neq 0$) is defined as $\sum_{i=1}^n a_i$, where $\alpha = (a_1, \dots, a_n)$. The (total) degree of $f$ is denoted by $\deg(f)$ and is defined as $\max(\sum_{i=1}^na_i\mid\alpha=(a_1,\dots,a_n),c_\alpha\neq 0)$. For $m\in\Z_{>0}$, sum of all terms of $f$ of the total degree equal to $m$ is called the $m$-th \emph{homogeneous component} of $f$. A polynomial is sum of its homogeneous components and a homogeneous polynomial has only one non-zero homogeneous component.

Given a fixed monomial order on $R$, say $\preceq$, then $x^{\alpha^\star}$ is the \emph{leading monomial} of $f$, denoted $\LM(f)$, if $\forall\alpha\in A$; $x^\alpha\preceq x^{\alpha^\star}$ and $c_{\alpha^\star}\neq 0$. 
The coefficient of $\LM(f)$ is called the \emph{leading coefficient} of $f$ and is denoted by $\LC(f)$. The \emph{leading term} of $f$ is denoted by $\LT(f)$ and is defined to be $\LC(f)\LM(f)$. For a set of polynomials $F$ we denote the ideal generated by $F$ in $R$ by $\langle F\rangle$. For an ideal $I\subseteq R$, the ideal generated by leading monomials of the polynomials of $I$ is called the \emph{initial ideal} of $I$ and is denoted by $\In(I)$, i.e. $\In( I ) = \langle \LM( f ) \mid f \in I \rangle$.

\subsection{Gr\"obner Bases and Buchberger's Algorithm}

As discussed in the introduction, one of the most important tools in computer algebra is the Gr\"obner basis, whose computation we aim to speed up. 

\begin{definition}\label{def:GB}
    A set $F\subseteq R$ is a \emph{Gr\"obner basis} for a fixed monomial order if $\langle \LM(f)\mid f\in F\rangle=\In(\langle F\rangle)$.
\end{definition}

For a finite set of polynomials $F\subset R$ and a fixed monomial order, we will denote the least common multiple of the leading monomials of elements of $F$ by $\el(F)$, that is $\el(F)=\lcm\big(\LM(f),f\in F\big)$.

For two polynomials $f,g\in R$ and a fixed monomial order, the \emph{S-polynomial} of $f$ and $g$, denoted by $\es(f,g)$, is defined as follows.
\[\es(f,g)=\frac{\el(f,g)}{\LT(f)}f-\frac{\el(f,g)}{\LT(g)}g.\]

A polynomial $f$ is \emph{reduced} with respect to $F\subseteq R$, for a fixed monomial order, if there exists no monomial in $f$ that is divisible by any leading monomial of polynomials in $F$. We can always reduce an arbitrary polynomial with respect to a set of polynomials, e.g. using (multivariate) polynomial division. 
If $f,f_1,\dots,f_r\in R$ and the remainder of division of $f$ with respect to $F=\lbrace f_1,\dots,f_r\rbrace$ is $g$, then $g$ is reduced with respect to $F$ and we say $f$ is reduced to $g$ under $F$, denoted $\overline{f}^F=g$.  We can now express, \emph{Buchberger's Criterion} for a Gr\"obner basis.

\begin{theorem}\label{thm:Buchberger_criteria}
    (\cite[Thm. 2.6.6]{Cox_et_al_2015}). 
    Let $F=\lbrace f_1,\dots,f_r\rbrace\subset R$ and fix a monomial order, then $F$ is a Gr\"obner basis if and only if for all pairs $i\neq j$, $\overline{\es(f_i,f_j)}^F=0$.
\end{theorem}

Buchberger's criterion provides a constructive approach, recalled in Algorithm \ref{alg:Buchberger_algorithm}, to extend any set of polynomials to a Gr\"obner basis. That is, simply calculate the S-polynomial of each pair and if it is not reducible to 0, add the remainder, $\overline{\es(f_i,f_j)}^F$, to the original set; continuing this process with the new set and repeating until S-polynomials of all pairs reduce to zero (see \cite[Thm. 2.7.2]{Cox_et_al_2015}).  

\begin{algorithm}[ht]
    \caption{Buchberger's algorithm to compute a Gr\"obner basis.  The sub-algorithm \q{PairsUpdate} updates the set of S-pair candidates after each step, and the sub-algorithm \q{NextPair} selects the next S-pair. We assume an implicit order among elements of $G$: the order in which they were added to the set.}
    \label{alg:Buchberger_algorithm}
    \begin{algorithmic}[l]
        \STATE{\textbf{Function} \q{Groebner}( $F$, $\preceq$ ):}
        \REQUIRE{$F=\lbrace f_1,\dots,f_r\rbrace\subset R$ and $\preceq$ a monomial order on $R$.}
        \ENSURE{$G$, a Gr\"obner basis for $\langle F\rangle$ with respect to $\preceq$.}
        \STATE $G=\lbrace f_1\rbrace$, $Q=\lbrace\,\rbrace$.
        \FOR{$i=2,\dots,r$}
            \STATE $G=G\cup\lbrace f_i\rbrace$.
            \STATE $Q=$ \q{PairsUpdate}($G$, $Q$, $\preceq$).
        \ENDFOR
        \WHILE{$Q\neq\lbrace\,\rbrace$}
            \STATE $(i,j)=$ \q{NextPair}($G$, $Q$, $\preceq$).
            \STATE $g=\overline{\es(G_i,G_j)}^{\, G}$.
            \STATE $Q=Q\setminus\lbrace (i,j)\rbrace$.
            \IF{$g\neq 0$}
                \STATE $G=G\cup\lbrace g\rbrace$.
                \STATE $Q=$ \q{PairsUpdate}($G$, $Q$, $\preceq$).
            \ENDIF
        \ENDWHILE
        \RETURN $G$.
    \end{algorithmic}
\end{algorithm}

\begin{algorithm}[ht]
    \caption{For use as \texttt{PairsUpdate} in Algorithm \ref{alg:Buchberger_algorithm} to update the set of S-pairs upon addition of a new polynomial to the basis. No reduction is performed.}
    \label{alg:Naive_PairsUpdate}
    \begin{algorithmic}[l]
        \STATE{\textbf{Function} \q{NaivePairsUpdate}( $F$, $Q$, $\preceq$ ):}
        \REQUIRE{$F\subset R$ is a set of polynomials with an implicit order, $Q\subseteq\lbrace (i,j)\mid 1\leq i\lneqq j\leq |F|-1\rbrace$, and $\preceq$ is a monomial order on $R$.}
        \ENSURE{$Q\subseteq\lbrace (i,j)\mid 1\leq i\lneqq j\leq |F|\rbrace$, an updated set of S-pair candidates for the next iteration of Buchberger's algorithm.}
        \STATE{\emph{initialisation:} Let $r=|F|-1$.}
        \STATE $Q=Q\cup\lbrace (i,r+1)\mid 1\leq i\leq r\rbrace$.
        \RETURN $Q$.
    \end{algorithmic}
\end{algorithm}

\subsection{Reduction Criteria}

We refer to those pairs we need to build S-polynomials for as \emph{S-pairs}.  The most expensive action in Algorithm \ref{alg:Buchberger_algorithm} is the polynomial division. If we can detect that some of these S-pairs would reduce to zero without actually calculating the S-polynomial and performing the polynomial division, it will save computation time. We recall, in Algorithms~\ref{alg:Basic_PairsUpdate} and \ref{alg:GebauerMoller_PairsUpdate}, two strategies from the literature for reducing the size of the set of S-pairs, essentially concluding, at little cost, that certain S-pairs will not result in any additions to the basis. 

\begin{algorithm}[ht]
    \caption{For use as \texttt{PairsUpdate} in Algorithm \ref{alg:Buchberger_algorithm} to update the set of S-pairs upon addition of a new polynomial to the basis.  Uses two basic criterion \cite[Proposition 2.10.1]{Cox_et_al_2015} and \cite[Proposition 2.10.8]{Cox_et_al_2015} to remove candidate S-pairs where it may be concluded that their S-polynomials will not add any new basis elements in future iterations of Buchberger's algorithm.}
    \label{alg:Basic_PairsUpdate}
    \begin{algorithmic}
        \STATE{\textbf{Function} \q{BasicPairsUpdate}( $F$, $Q$, $\preceq$ ):}
        \REQUIRE{$F\subset R$ is a set of polynomials with an implicit order, $Q\subseteq\lbrace (i,j)\mid 1\leq i\lneqq j\leq |F|-1\rbrace$, and $\preceq$ is a monomial order on $R$.}
        \ENSURE{$Q'\subseteq\lbrace (i,j)\mid 1\leq i\lneqq j\leq |F|\rbrace$, an updated set of S-pair candidates for the next iteration of Buchberger's algorithm.}
        \STATE{\emph{initialisation:} Let $r=|F|-1$ and denote elements of $F$ by $f_1$, ..., $f_r$ and $g$ (the $(r+1)$-th element).}
        \FOR{$i=1,\dots,r$}
            \IF{$\el(f_i,g)\neq\LM(f_i)\LM(g)$}
                \STATE $Q=Q\cup\lbrace (i,r+1)\rbrace$.
            \ENDIF
        \ENDFOR
        \STATE $Q'=Q$.
        \FOR{$(i,j)\in Q$}
            \IF{$\el(f_i,f_j,f_k)=\el(f_i,f_j)$ and $(i,k),(k,i),(j,k),(k,j)\not\in Q'$}
                \STATE $Q'=Q'\setminus\lbrace (i,j)\rbrace$.
            \ENDIF
        \ENDFOR
        \RETURN $Q'$.
    \end{algorithmic}
\end{algorithm}

\begin{algorithm}[ht]
    \caption{For use as \texttt{PairsUpdate} in Algorithm \ref{alg:Buchberger_algorithm} to update the set of S-pairs upon addition of a new polynomial to the basis.  Uses the three criteria of \cite{Gebauer_Moller_1988} to remove candidate S-pairs where it may be concluded that their S-polynomials will not add any new basis elements in future iterations of Buchberger's algorithm.
    }
    \label{alg:GebauerMoller_PairsUpdate}
    \begin{algorithmic}
        \STATE{\textbf{Function} \q{GebauerMollerPairsUpdate}( $F$, $Q$, $\preceq$ ):}
        \REQUIRE{$F\subset R$ is a set of polynomials with an implicit order, $Q\subseteq\lbrace (i,j)\mid 1\leq i\lneqq j\leq |F|-1\rbrace$, and $\preceq$ is a monomial order on $R$.}
        \ENSURE{$Q\subseteq\lbrace (i,j)\mid 1\leq i\lneqq j\leq |F|\rbrace$, an updated set of S-pair candidates for the next iteration of Buchberger's algorithm.}
        \STATE{\emph{initialisation:} Let $r=|F|-1$ and denote elements of $F$ by $f_1$, ..., $f_r$ and $g$ (the $(r+1)$-th element).}
        \STATE $Q'=\lbrace\,\rbrace$.
        \FOR{$(i,j)\in Q$}
            \IF{ \textbf{not} \big( $\el(f_i,f_j,g)=\el(f_i,f_j)$, $\el(f_i,f_j)\neq\el(f_i,g)$, $\el(f_i,f_j)\neq\el(f_j,g)$ and $\el(f_i,g)\neq\el(f_j,g)$ \big)}
                \STATE $Q'=Q'\cup\lbrace (i,j)\rbrace$.
            \ENDIF
        \ENDFOR
        \STATE $Q''=\lbrace\,\rbrace$.
        \FOR{$i=1,\dots,r$}
            \STATE Label the pair $(i,r+1)$ with $\el(f_i,g)$.
            \IF{no pair in $Q''$ has the same label with $(i,r+1)$}
                \STATE $Q''=Q''\cup\lbrace (i,r+1)\rbrace$.
            \ELSIF{$\el(f_i,g)=\LM(f_i)\LM(g)$}
                \STATE replace the pair in $Q''$ that has the same label with $(i,r+1)$.
            \ENDIF
        \ENDFOR 
        \STATE $Q''=Q''\setminus\lbrace (i,r+1)\mid \el(f_i,g)=\LM(f_i)\LM(g)\rbrace$.
        \STATE Delete all $(i,r+1)$ from $Q''$ such that there exists $(j,r+1)$ in $Q''$ with $i\neq j$ and $\el(f_j,g)$ divides $\el(f_i,g)$.
        \STATE $Q=Q'\cup Q''$.
        \RETURN $Q$.
    \end{algorithmic}        
\end{algorithm}

\begin{theorem}\label{thm:Buchberger_algorithm}
    Let $F\subset R$ and fix a monomial order. A Gr\"obner basis for $\langle F\rangle$ can be constructed in finite steps by Algorithm~\ref{alg:Buchberger_algorithm}.
\end{theorem}

\begin{proof}
    Algorithm~\ref{alg:Buchberger_algorithm} when \q{PairsUpdate} is placed by Algorithm~\ref{alg:Naive_PairsUpdate}, Algorithm~\ref{alg:Basic_PairsUpdate} or Algorithm~\ref{alg:GebauerMoller_PairsUpdate} is the same as the algorithms in \cite[Theorem 2.7.2]{Cox_et_al_2015}, \cite[Subsection 4.2]{Gebauer_Moller_1988} (also \cite[Theorem 2.10.9]{Cox_et_al_2015}) or \cite[Subsection 4.4]{Gebauer_Moller_1988} respectively, and so we refer the reader to these sources for proofs of correctness.
\end{proof}

\subsection{Classical \textit{S}-Pair Selection Strategies}
\label{sec:S-pair_selections}

At every iteration, while the S-pair queue is non-empty, there is a question of which S-pair candidate to check next? Note that the number of iterations of the algorithm is not fixed and is impacted by this decision.  Hence a variety of strategies have been developed on selection of the next S-pair. 

In Buchberger's algorithm with a naive S-pairs set update strategy (Algorithm~\ref{alg:Naive_PairsUpdate}), S-polynomials of all S-pairs would be calculated at least once independently of the S-pair selection strategy.  However, when basic or Gebauer-M\"oller update strategies are used (Algorithms~\ref{alg:Basic_PairsUpdate} and \ref{alg:GebauerMoller_PairsUpdate}) some of the S-pairs that had been added in previous iterations and had not been processed may be deleted from future iterations without being processed at any stage. Therefore an S-pair selection strategy that can make good decisions without calculating S-polynomials can save some computation time. That is, if we make the decision using information from the calculation of S-polynomials of all current possible S-pairs then some of those calculations would not be used later and thus are a cost in making the decision itself.  Hence most classical strategies have focussed on information derived from the S-pair rather than S-polynomial.

\subsubsection{The First Strategy.}

This trivial ``strategy'' treats the set of S-pairs as a queue, selecting to process next whichever pair had been added to the set earliest by the algorithm.

\subsubsection{The Normal Strategy.}

Let $F\subset R$ and $I=\langle F\rangle$. If $\lbrace\LM(f)\mid f\in F\rbrace$ generates $\In(I)$, then it is a Gr\"obner basis. The initial ideal, $\In(I)$, is a monomial ideal and a monomial ideal has a unique minimal basis consisting of monomials, none of which divide the other ones \cite[Proposition 2.4.7]{Cox_et_al_2015}. Buchberger's algorithm will eventually find elements of $I$ that their leading monomials are the elements of this minimal basis of $\In(I)$. Therefore, the S-pair for which the reduced version of its S-polynomial is non-zero and has a minimal leading monomial among other cases with respect to divisibility, is sensible choice for the next iteration.

Let us replace the minimality with respect to divisibility by being the smallest with respect to the monomial order. Let $f,g\in F$; then as a direct consequence of the definitions of polynomial division and definition of S-polynomials, we have $\LM\big(\overline{\es(f,g)}^F\big)\preceq\LM\big(\es(f,g)\big)$, and $\LM\big(\es(f,g)\big)\preceq\el(f,g)$ respectively:
\begin{equation}\label{eq:normal_inequality}
\LM\big(\overline{\es(f,g)}^F\big)\preceq\LM\big(\es(f,g)\big)\preceq\el(f,g).
\end{equation}
Therefore to avoid performing the polynomial division, which is the costly part of Buchberger's algorithm, we can use $\el(f,g)$ as an upper-bound to control and judge which S-pair leads with a new polynomial of the lowest leading monomial to be added. This strategy is called the \emph{normal} strategy \cite{Buchberger_1979}. 
In the case of a tie between two S-pairs, the tie-breaking is done by the First Strategy.

\subsubsection{The Degree Strategy.}

Independently of the monomial order, a monomial of lower degree divides a larger set of monomials than a monomial of a higher degree: that means the possibility of reducing more S-pairs to zero and getting closer to the termination of Buchberger's algorithm. Again we want to avoid performing polynomial division, therefore we are more interested in using an upper-bound. A monomial order is called \emph{degree-compatible} if it first compares degrees of two monomials and declares the one with the lower degree is the smaller monomial, then it breaks the tie using other comparisons. Graded lexicographic, \emph{grlex}, and graded reverse lexicographic, \emph{grevlex}, are examples of degree-compatible monomial orders, while lexicographic order, \emph{lex}, is not. If $\preceq$ is a degree-compatible monomial order, then for any polynomial $h\in R$ we have $\deg(h)=\deg\big(\LM(h)\big)$, thus \eqref{eq:normal_inequality} implies $\deg\big(\overline{\es(f,g)}^F\big)\leq\deg\big(\es(f,g)\big)\leq\deg\big(\el(f,g)\big)$. The strategy of picking the S-pair whose $\el(f,g)$ has the lowest degree and then breaking ties with first, is called the \emph{degree} strategy \cite{PSH_20}. 

\subsubsection{The Sugar Strategy.}

For any two polynomials $f,g\in R$, $\frac{\el(f,g)}{\LT(f)}$ is a scalar times a monomial with degree equal to $\deg\big(\el(f,g)\big)-\deg\big(\LM(f)\big)$. Using the two relations that $\deg(f+g)\leq\max(\deg(f),\deg(g))$ and $\deg(fg)=\deg(f)+\deg(g)$, one can get an upper bound for degree of S-polynomials.
\begin{multline}
    \label{eq:sugar_inequality}
    \deg\big(\es(f,g)\big)\leq\max\Big(\deg\big(\el(f,g)\big)-\deg\big(\LM(f)\big)+\deg(f),\\ \deg\big(\el(f,g)\big)-\deg\big(\LM(g)\big)+\deg(g)\Big)
\end{multline}
The right hand side of \eqref{eq:sugar_inequality} is named the \emph{sugar} of the S-pair $(f,g)$ and is denoted by $\sugar(f,g)$. When $\preceq$ is a degree-compatible monomial order or $f$ and $g$ are homogeneous polynomials we have $\sugar(f,g)=\deg\big(\el(f,g)\big)$. The sugar strategy for S-pair selection chooses a pair with lower sugar, and breaks the tie with Normal strategy \cite{Giovini_et_al_1991}.

\section{The New Homogeneity Entropy Selection Strategy}
\label{sec:entropy}

Let $X$ be a discrete random variable with a finite sample space $\Omega=\lbrace s_1,\dots,s_t\rbrace$ and (density) probability distribution $p(X=s_i)=p_i$. Then the (Shannon) entropy of $X$ is $-\sum_{i=1}^tp_i\log\big(p_i\big)$ where $\log$ is the logarithm in base 2. 

\begin{definition}
    \label{def:homogeneity}
    Let $f=\sum_{\alpha\in A}c_\alpha x^\alpha$, where $A\subset\Z_{\geq 0}^n$ and $\forall\alpha\in A$, $c_\alpha\neq 0$, be a polynomial in $R$. Let $\Omega_f=\lbrace\deg(x^\alpha)\mid\alpha\in A\rbrace$. We define a random variable for $f$ that describes the chance of a randomly selected monomial of $f$ being of a given degree. The density probability function of this random variable assigns the proportion of the size of the homogeneous components to the size of the whole polynomial when size is measured by the number of monomials. That means
    \begin{equation}
        \label{eq:homogeneity_distribution}
        \forall i\in\Omega_f\,\colon\,p(X=i)=\frac{|\lbrace\alpha\in A\mid \deg(x^\alpha)=i\rbrace|}{|A|}.
    \end{equation}
    We define the \emph{homogeneity entropy} of $f$ to be the entropy of this random variable. We denote this entropy by $\E_f$ and the probability distribution in \eqref{eq:homogeneity_distribution} by $p_f$.
\end{definition}

\begin{lemma}
    \label{lem:entropy}
    \begin{itemize}
        \item[a.] If $f$ is a homogeneous polynomial, then $\E_f=0$.
        \item[b.] If $f$ is a binomial, then $\E_f$ is either 0, if it is homogeneous, or 1 otherwise.
        \item[c.] For any polynomial $f$, we have $0\leq\E_f\leq \log_2 |\Omega_f|$.
    \end{itemize}
\end{lemma}

\begin{proof}
    The proofs of (a) and (b) are trivial; (c) follows from the fact that maximum entropy occurs when the probability distribution is uniform \cite[p.~394]{Shannon_1948}.
\end{proof}

We define a new \emph{homogeneity entropy} S-pair selection strategy to choose the pair for which the corresponding S-polynomial has the lowest entropy, breaking any ties with the First strategy. We refer to this as the Entropy strategy for short. This strategy will prioritise S-pairs such that their S-polynomials are homogeneous, and avoid those with many homogeneous components. 

When an S-polynomial is a monomial (but not reduced to 0), it is a good choice, because in later polynomial divisions it simply removes all monomials that are a multiple of it.  Correspondingly, when the S-polynomial has many terms, the later polynomial divisions will be harder. We may think of the entropy strategy as an attempt to increase the chance of picking monomials and avoid large polynomials.  However, this is not strictly the case as, for example, it does not differentiate a monomial from a large homogeneous polynomial. 
 
One key difference between this new strategy and the classic ones summarised in Section \ref{sec:S-pair_selections} is that the new strategy requires the calculation of the S-polynomial to make the decision, while the others do not.  Of course, this bears an extra expense in the selection step. In the case of using Buchberger algorithm (Algorithm~\ref{alg:Buchberger_algorithm}) without using S-pairs reduction strategies (such as Algorithms~\ref{alg:Basic_PairsUpdate} and \ref{alg:GebauerMoller_PairsUpdate}) this extra expense is one that would have been made later in the algorithm anyway, but that is not the case with the more advanced strategies.  So this extra cost needs to be weighed against the value of the decisions being made. 

Note that making the selection based on the S-polynomial rather than the S-pair is something we experimented with following the results of \cite[\S 5.1]{PSH_20} (written up in more detail in the thesis \cite[\S 5.4.3]{Peifer2021}) which attempted to interpret the selection strategy of an agent trained by reinforcement learning and observed preferences based on properties of the S-polynomial rather than S-pair.

In the remainder of this paper, we experimentally compare the performance of this new selection strategy, with those in the literature.


\section{Experimental Methodology}

\subsection{Datasets}

\subsubsection{Synthetic Datasets.}

We have synthetically generated a range of datasets following the data generation process used in \cite{PSH_20}.  Each dataset consists of 1000 randomly generated polynomials defined over a fixed coefficient field $\mathbb{Z}/32003\mathbb{Z}$. These datasets are defined by four different parameters denoted \((v,\,d,\,p,\,\sigma)\), where \(v\) is the number of variables, \(d\) is the maximum total degree, \(p\) is the number of polynomials per system, and  \(\sigma\) is a density parameter used to determine the number of terms generated for each polynomial. 

We started with base dataset \texttt{3-20-10-0.3} meaning polynomial systems with 3 variables, maximum total degree 20, and exactly 10 input polynomials. For each polynomial, after generating a binomial, the number of additional terms is sampled from a Poisson distribution with parameter $\lambda=0.3$ controlling the polynomial sparsity. 
These parameters were varied to make further datasets:
\begin{itemize}
    \item dataset \texttt{3'5-20-10-0.3} has a higher number of variables compared to the base dataset, uniformly distributed between $3-5$; 
    \item dataset \texttt{3-20-10-0.5} increased the number of terms in each polynomial;
    \item dataset \texttt{3-25-10-0.3} has a higher maximum total degree of 25;
    \item dataset \texttt{3-10-10-0.3} has a lower maximum degree of 10; and
    \item dataset \texttt{3-20-2'10-0.3} has a lower number of input polynomials compared to the base dataset, varying the number between $2-10$. 
\end{itemize}


\subsubsection{PHCpack Benchmark Set.} 

We also experiment on a dataset of 94 polynomial systems drawn from the PHCpack benchmark collection, which aggregates systems collected from journal articles and real-world applications. Systems were imported into $\mathbb{Q}$ by exact decimal-to-rational conversion, with Buchberger’s algorithm then run over $\mathbb{Q}$. This dataset is described in \cite{VPHCpack} and is hosted online\footnote{\url{https://homepages.math.uic.edu/~jan/demo.html}} where it has grown since that original description. We discarded the systems with complex coefficients. 

\subsection{Buchberger's Algorithm Setup}

All experiments use Buchberger’s algorithm to compute reduced Gr\"{o}bner bases with respect to the graded reverse lexicographic monomial order (grevlex).

We apply the Gebauer–M\"{o}ller reduction criteria and inter-reduce the basis after each insertion, so that the different S-pair selection strategies are compared within an optimised, but fixed, Buchberger's algorithm implementation: \ref{alg:Buchberger_algorithm} using Algorithm \ref{alg:GebauerMoller_PairsUpdate}. The whole algorithm, along with the selection strategies, has been implemented in Python~3.7.16 with SymPy~1.10.1 \cite{SymPy2017}: this was a custom implementation because the implementation of Buchberger's algorithm that comes with SymPy does not give us sufficient freedom over the selection strategy to allow this experimentation.

\subsection{Runtime and Performance Metrics}

All experiments are run sequentially on a single core of a 12th-generation Intel Core i7-12700H (up to 4.7~GHz, 14 physical cores / 20 threads) with 32~GB RAM, and we report wall-clock runtime in seconds for each Gr\"{o}bner basis computation. All the reported runtimes include the cost of heuristic use. S-polynomials calculated by the Entropy heuristic are cached and reused later if needed.  Each run is terminated if it exceeds a timeout of 10 minutes (600 seconds) and such runs are marked as timeouts (TO).   

For each dataset and selection strategy, we report summary statistics around the runtime achieved.  We report on two aggregations:  one for the full dataset and one for just those examples for which all strategies complete before timeout.  In the former, the runtime of timeouts is included as the timeout limit (600s) and thus is a lower bound on the true runtime.  For the latter we also report the mean and median average runtime, and also the 95th percentile (only for experiments on PHCpack where a small number of problems in this dataset exhibit significantly larger runtimes than the bulk of dataset).

In addition we report the number of problems for which a strategy timed out and the number for which it achieved the minimum observed time.  Finally, we report the ratio of a strategy's runtime versus the runtime with the Entropy strategy (computed only on the subset of systems completed by all strategies) to understand the level of improvements Entropy achieves.

\section{Experimental Results}

In this section we evaluate the different S-pair selection strategies introduced above, keeping all other hardware and software components the same.  

\subsection{Random Polynomial Systems}
\label{section:Random_results}

Table~\ref{tab:random-results} reports the results on the various synthetic datasets.


Where all three strategies terminate on every instance, the Entropy strategy consistently yields significantly smaller total runtimes than Sugar, Normal and Degree, often by an order of magnitude: Sugar takes between 3.64--8.8 times longer than Entropy, Normal takes between 3.47--11.83 times longer, and Degree takes 3.19--12.4 times longer. Every distribution contains significant subsets on which each of the strategies performs best, suggesting each has separate useful information to draw on.  These subsets are much larger for Entropy: between 421--672 of the systems within each distribution, against Degree's 151--283, Normal's 70--158 and Sugar's 54--154.  Degree and Normal seem to consistently have larger such subsets than Sugar, while taking longer on the distribution as a whole.

All the strategies completed on all the instances for base dataset \texttt{3-20-10-0.3} with only a handful of timeouts on the other datasets except for \texttt{3'5-20-10-0.3}.  This is the dataset where we include more variables, the dominant parameter in the theoretical complexity of Gr\"{o}bner basis computation, and so it is not surprising that this dataset exhibits many more timeouts.  Note that Degree, Normal and Sugar exhibited many more timeouts compared to Entropy here (41, 37 and 59 respectively, compared to 8).  

Now let us consider the impact of the other dataset parameters.
\begin{itemize}
\item The effect of varying the second parameter, the maximum degree, can be observed by comparing the first three distribution rows:  the statistic in the final column shows the ratios of runtime for Degree, Normal and Sugar compared to Entropy increasing with the degree, indicating that the Entropy strategy is more robust to growth in the degree parameter. 
\item The effect of the third parameter, the number of input polynomials, can be observed by comparing distribution rows 1 and 5:  where again we see the ratios increasing with the parameter, indicating that the Entropy strategy is more robust to growth in the system size.
\item The effect of the final parameter, the sparsity measure in the polynomials, can be observed by comparing distribution rows 1 and 4:  it seems the benefit of entropy weakens as the polynomials become denser.
\end{itemize}
\begin{table}
\centering
\caption{Performance on random polynomial systems 
(1000 instances per synthetic configuration; 600\,s per-instance timeout).
\#TO is the number of timeouts; \# min Time counts instances on which
the strategy was the strict wall-clock winner among all the strategies (ties excluded); ``Runtime incl. TO'' aggregates wall-clock time with timeouts recorded at the 600\,s limit; ``Runtime w/o TO'' aggregates only over instances solved by all the strategies, for which we also report the Mean and Median statistics and Time vs Entropy which is the ratio of the runtime against what Entropy achieves.  The best performance for each metric on each dataset is in \textbf{\textcolor{blue}{bold blue text}}.}
\label{tab:random-results}
\begin{tabular}{l l r c r r r r c}
\toprule
Dataset & Strategy & \#TO & \shortstack{\# min\\Time}
  & \shortstack{Runtime\\incl. TO} & \shortstack{Time\\w/o TO}
  & \shortstack{Mean\\w/o TO} & \shortstack{Median\\w/o TO}
  & \shortstack{Time vs \\ Entropy} \\
\midrule
\texttt{\textbf{3-20-}}  & Sugar   & 0 & 72                        & 13m 12.3s                       & 13m 12.3s                       & 0.79\,s                       & 0.020\,s                       & 8.80 \\
\, \texttt{\textbf{10-0.3}} & Normal   & 0 &   125                    &    17m 45s   &   17m 45s   &    1.065\,s               &                  0.011\,s      & 11.83 \\
                               & Degree  & 0 & 172 & 18m 36.6s                       & 18m 36.6s                       & 1.12\,s                       & 0.012\,s & 12.41 \\
                               & Entropy & 0 & {\color{blue}\textbf{631}}                        & {\color{blue}\textbf{1m 30s}} & {\color{blue}\textbf{1m 30s}} & {\color{blue}\textbf{0.083\,s}}& {\color{blue}\textbf{0.007\,s}}                       & {\color{blue}\textbf{1.00}} \\
\midrule
\texttt{\textbf{3-10-}}  & Sugar   & {\color{blue}\textbf{0}}                        & 54                        & 47m 17.6s                        & 46m 22.6s                        & 2.84\,s                       & 0.100\,s                       & 3.88 \\
\, \texttt{\textbf{10-0.3}} & Normal   & 1 &  106                      &       1h 00:36s      &    50m 36s                   &                     3.04\,s   &     0.047\,s    & 4.23 \\
                               & Degree  & 1                        & 167 & 1h 00:045s                       & 50m 5s                        & 3.01\,s                       & 0.043\,s                       & 4.19 \\
                               & Entropy & {\color{blue}\textbf{0}} & {\color{blue}\textbf{672}}                        & {\color{blue}\textbf{11m 58s}} & {\color{blue}\textbf{11m 58s}} & {\color{blue}\textbf{0.72\,s}}& {\color{blue}\textbf{0.009\,s}}& {\color{blue}\textbf{1.00}} \\
\midrule
\texttt{\textbf{3-25-}}  & Sugar   & {\color{blue}\textbf{0}}                      & 63                       & 54m 38.2s                       & 43m 14.4s                       & 3.28\,s                       & 0.043\,s                        & 4.88 \\
\, \texttt{\textbf{10-0.3}} & Normal   & 1 &     158                   &    1h 38m 35s     &   1h 20m 7s   &  4.822\,s                      &          0.019\,s              & 9.04 \\
                               & Degree  & 3                        & 151 & 1h 59m 55s                      & 1h 29m 55s                     & 5.41\,s                       & 0.018\,s & 10.14 \\
                               & Entropy & {\color{blue}\textbf{0}} & {\color{blue}\textbf{625}}                       & {\color{blue}\textbf{8m 52s}} & {\color{blue}\textbf{8m 52s}} & {\color{blue}\textbf{0.534\,s}}& {\color{blue}\textbf{0.009\,s }}                       & {\color{blue}\textbf{1.00}} \\
\midrule
\texttt{\textbf{3-20-}}  & Sugar   & {\color{blue}\textbf{0}} & 54                        & 52m 11.6s                        & 51m 08.2s                        & 3.13\,s                       & 0.114\,s                        & 4.05 \\
\, \texttt{\textbf{10-0.5}} & Normal   & 1 &  144                      &  1h 0m 34s   &             50m 34s      &  3.037\,s                      &      0.048\,s                  & 4.00 \\
                               & Degree  & 1                        & 137                      & 1h 07m 38s                       & 57m 38.6s                        & 3.46\,s                       & 0.050\,s                        & 4.56 \\
                               & Entropy & {\color{blue}\textbf{0}} & {\color{blue}\textbf{664}} & {\color{blue}\textbf{12m 38s}} & {\color{blue}\textbf{12m 38s}} & {\color{blue}\textbf{0.76\,s}}& {\color{blue}\textbf{0.01\,s}} &  {\color{blue}\textbf{1.00}}\\
\midrule
\texttt{\textbf{3-20-}}& Sugar   & 0 & 127                       & 45m 41.5s                        & 45m 41.5s                        & 2.74\,s                       & 0.016\,s                       & 3.64 \\
\, \texttt{\textbf{2'10-0.3}}& Normal   & 0 &   152                     & 1h 8m 35s                        & 1h 8m 35s                        &     0.753\,s                   &  0.007\,s                      &  5.46 \\

                               & Degree  & 0 & 268 & 57m 24.6s                        & 57m 24.6s                        & 3.44\,s                       & 0.010\,s & 4.57 \\
                               & Entropy & 0 & {\color{blue}\textbf{453}}                        & {\color{blue}\textbf{12m 33s}} & {\color{blue}\textbf{12m 33s}} & {\color{blue}\textbf{0.753\,s}}& {\color{blue}\textbf{0.007}}\,s                       & {\color{blue}\textbf{1.00}} \\

\midrule
\texttt{\textbf{3'5-20-}}& Sugar   & 59                        &154                        & 13h 32m 38s                      & 3h 33m 10s      & 13.78\,s       & 0.030\,s                        & 6.6 \\
\, \texttt{\textbf{10-0.3}} & Normal   & 37 &   70                     &    9h 20m 9s    & 1h 52m   &   7.24\,s                     &           0.034\,s             &  3.47 \\
                               & Degree  & 41                        & 283 & 9h 34m 42s                       & 1h 43m 14s                       & 6.67\,s                       & 0.020\,s & 3.19 \\
                               & Entropy & {\color{blue}\textbf{8}} & {\color{blue}\textbf{421}}                        & {\color{blue}\textbf{2h 25m}} & {\color{blue}\textbf{32m 19s}} & {\color{blue}\textbf{2.09\,s}} & {\color{blue}\textbf{0.014\,s}}                       & {\color{blue}\textbf{1.00}} \\

\bottomrule
\end{tabular}
\end{table}

The above suggests strong evidence for the performance of Entropy in comparison to the established heuristics.

\subsection{PHCpack Dataset}
\label{section:PHC_results}

Finally, we test on the PHCpack dataset, whose problems are sourced from the real world, with the overall results reported in Table~\ref{tab:phc-results}. This lists the number of timeouts and the two aggregated runtimes: total runtime on all \(94\) systems (including the time spent on timeouts), and the total runtime for the \(55\) systems that terminated without timing out for all three strategies. Table~\ref{tab:phc-results} also reports the statistics on mean and median, along with the number of systems that each strategy performed the best for.

The PHCpack dataset is different from the synthetic dataset qualitatively as it includes a mixture of problems that are structurally trivial systems (solvable in
milliseconds) and genuinely hard systems (requiring minutes to hours). All strategies exhibit some timeouts and so we focus attention to the 55 systems that were solved by all strategies. 
and we discuss them in turn.

\subsubsection{Aggregate Performance.}

The results in Table \ref{tab:phc-results} indicate that Sugar completes on the highest number of problems thus exhibiting the least number of timeouts, one more than Degree in second place, which solves three more than Normal and four more than Entropy. Because this dataset is comparatively small, the extra timeouts are far greater than the differences observed in the Runtime incl. TO column, and so a better idea of the comparative runtimes can be given by the aggregation excluding timeouts.   Sugar is the fastest on aggregate, as well as on the commonly solved systems than for the synthetic data in Table \ref{tab:random-results}.  

\begin{table}
\centering
\caption{Performance on the PHCpack benchmark suite (94 real-world polynomial systems drawn from engineering and scientific applications; 600\,s per-instance
timeout). \#TO is the number of timeouts.
``Runtime incl. TO'' aggregates wall-clock time with timeouts charged at the 600\,s limit; ``Runtime w/o TO'' aggregates only over the 55 instances solved by all three strategies. Mean and median are computed over each strategy's own non-timed-out runs. The best performance of each metric is in \textbf{\textcolor{blue}{blue bold text}}.}
\label{tab:phc-results}
\begin{tabular}{l c c r r r c}
\toprule
Strategy & \#TO
  & \shortstack{Runtime\\incl. TO} & \shortstack{Runtime\\w/o TO}
  & \shortstack{Mean\\w/o TO} & \shortstack{Median\\w/o TO} & \shortstack{\# of wins}
\\
\midrule
Normal  & 37                       
    & 6h 26m 15s                        
    & 14m 3s                        
    &15.32\,s
    & 0.155\,s  
    & 14 \\
Sugar   & {\color{blue}\textbf{33}} 
    & {\color{blue}\textbf{5h 41m 22s}} 
    & {\color{blue}\textbf{3m 47s}}                      
    & {\color{blue}\textbf{4.13\,s}}                        
    & {\color{blue}\textbf{0.121\,s}}
    & 10\\
Degree  & 34                       
    & 6h 8m 30s                        
    & 8m 4s                        
    &8.79\,s
    & 0.131\,s   
    & {\color{blue}\textbf{25}} \\
Entropy & 38                        
    & 6h 31m 1s                        
    & 10m 59s
    & 11.97\,s                        
    & 0.150\,s 
    & 6 \\
\bottomrule
\end{tabular}
\end{table}

Normal attains the highest mean per-system time, 15.32s, on completions, and a highest median of 0.155s. The overall runtime and corresponding number of wins for Degree indicate it has a long tail of slow runs combined with many trivially fast ones. Entropy attains a mean of 11.97s which is between Normal and Degree. Sugar ends up with the lowest number of timeouts and lowest aggregate runtime thus indicating its supremacy on these systems.

\begin{table}[t]
\centering
\caption{PHCpack benchmark, stratified by difficulty using sugar's wall time. Per-bucket best in \textbf{\textcolor{blue}{blue bold text}}. Wins = strict time wins (unique fastest) within bucket.}
\label{tab:phcpack-stratified}
\begin{tabular}{l l c c c c}
\toprule
Difficulty bucket & Strategy  & Mean & Median & Total & Wins \\
\midrule
Easy ($<0.1$\,s, $n=25$) & Sugar & 26\,ms & 32\,ms & 649\,ms & 2 \\
 & Degree & {\color{blue}\textbf{25\,ms}} & {\color{blue}\textbf{8\,ms}} & {\color{blue}\textbf{616\,ms}} & {\color{blue}\textbf{18}} \\
 & Normal & 35\,ms & 17\,ms & 866\,ms & 4 \\
 & Entropy & 39\,ms & 26\,ms & 977\,ms & 1 \\
\midrule
Medium ($0.1$--$10$\,s, $n=21$) & Sugar & {\color{blue}\textbf{1.28\,s}} & 533\,ms & {\color{blue}\textbf{26.82\,s}} & 5 \\
 & Degree & 7.35\,s & {\color{blue}\textbf{454\,ms}} & 2m 34.4s & 5 \\
 & Normal & 5.32\,s & 1.34\,s & 1m 51.7s & {\color{blue}\textbf{6}} \\
 & Entropy & 2.43\,s & 691\,ms & 51.02\,s & 5 \\
\midrule
Hard ($>10$\,s, $n=9$) & Sugar & {\color{blue}\textbf{22.21\,s}} & 18.00\,s & {\color{blue}\textbf{3m 19.9s}} & 3 \\
 & Degree & 36.50\,s & 20.38\,s & 5m 28.5s & 2 \\
 & Normal & 1m 21.1s & {\color{blue}\textbf{15.13\,s}} & 12m 10.1s & {\color{blue}\textbf{4}} \\
 & Entropy & 1m 7.4s & 31.26\,s & 10m 6.7s & 0 \\
\bottomrule
\end{tabular}
\end{table}

\subsubsection{Stratified Performance Analysis.}

The aggregate results do not reveal any patterns of which systems each strategy performs better for and so we generated Table \ref{tab:phcpack-stratified} which classifies the systems according to the runtime of Sugar as a baseline, for a proxy for difficulty.
From the table it is evident that Degree performs best for trivial systems that can be solved under 0.1s completing fastest on most of the systems with Entropy worst, most likely because of its greater overhead in the need to compute all S-polynomials dominating at sub-millisecond timescales. Sugar performs the best on medium systems that complete between 0.1s--10s but Entropy is competitive, taking a non-trivial share of strict wins and does better than Degree and Normal. The gap between Sugar and Degree in terms of total time in this bucket demonstrates that Degree's cost advantage on easy systems does not survive into this tier. Sugar outperforms the others on systems that take higher than 10s. 

\subsubsection{System-level Runtime Distribution Analysis.}

To understand the performance of entropy we plotted a scatter plot between the runtimes for Sugar vs Entropy, Normal vs Entropy and Degree vs Entropy as shown in Figure \ref{fig:scatter} across the 55 commonly-solved systems. Against Sugar (left), most points either lie on or just below the diagonal. There is a small systematic downward bias indicating Sugar's better performance however it also shows that Entropy does not perform catastrophically worse than Sugar. 

Against Degree (centre), the points are roughly balanced around the diagonal, with Entropy clearly faster on a handful of systems (the blue points in the upper-left) and clearly slower on some systems (the red points in the lower-right). The presence of distinct off-diagonal outliers in both panels, rather than a uniform shift, shows that the aggregate timing gap is driven by a small number of pathological systems rather than overall overhead.

Against Normal we see an interesting trend where we see a non-trivial set of wins (blue) in upper left corner. We also observe a spread across the diagonal similar to Sugar and Degree. 

\begin{figure}[htbp]
    \centering
    \includegraphics[width=1\linewidth]{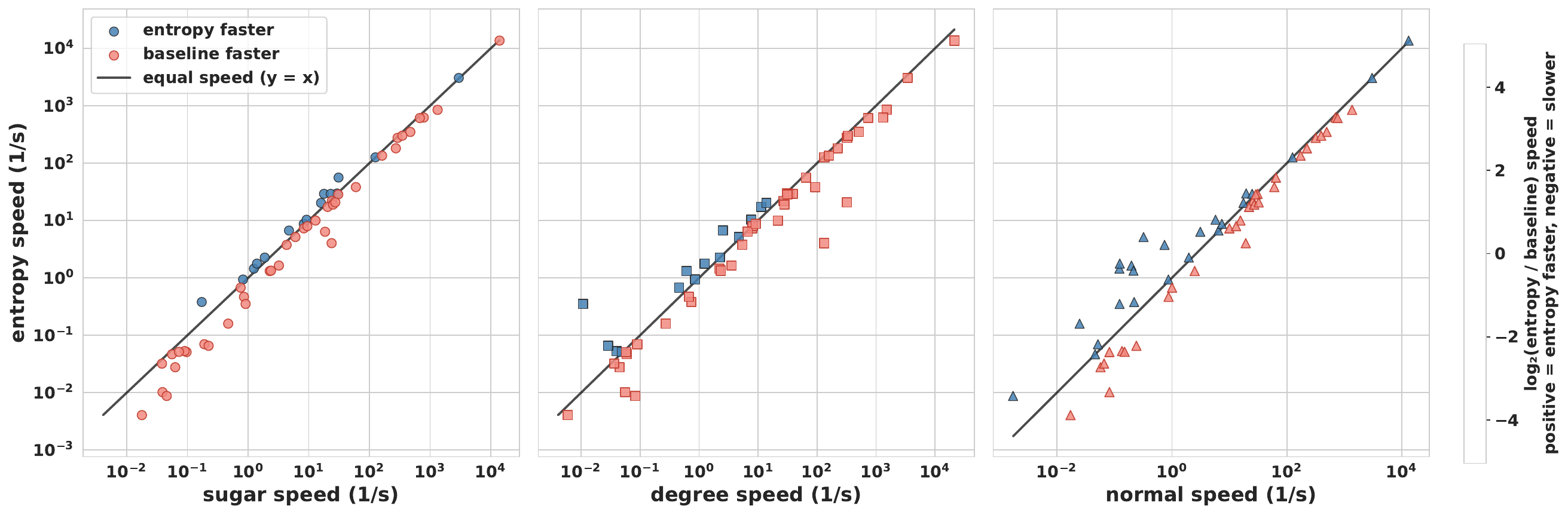}
    \caption{Per-system wall-clock comparison on the 55 commonly-solved PHCpack systems, with Entropy on the vertical axis against Sugar (left), Degree (centre) and Normal (right). Both axes are log-scaled. Points above the diagonal are systems where Entropy is faster; points below are systems where it is slower.}
    \label{fig:scatter}
\end{figure}

This concentration of higher runtimes in a few systems as shown in Figure \ref{fig:scatter} suggests Entropy's failure on PHCpack is not generic across the dataset. Three systems account for approximately 78\% of Entropy's total-time deficit relative to Sugar.  We also observe that Entropy ranks many candidate pairs as equivalent, breaking ties with First to produce poor choices; which happens far less for Sugar thus Sugar benefits from the Normal tie-break by definition.  This would suggest that while Entropy has useful information it needs to be combined with additional information for the systems in this dataset.

\subsection{Experimenting with a Synthetic PHCpack-shaped Dataset}
\label{sec:PHCpackShaped}

The experiments of the previous two subsections delivered radically different results:  Entropy dominated on the synthetic datasets produced with the methodology of \cite{PSH_20} but did not perform as well as the established heuristics on the PHCpack dataset of \cite{VPHCpack}.  We wanted to understand if this is explained by the different ``shape'' of the PHCpack dataset, or represented something deeper about the structure of polynomials emerging from examples.

We thus synthetically generated one further dataset: a new set of 100 problems, each produced from a template system from the PHCpack dataset.  The new polynomials synthesised used the same number of monomials and the same maximum degree as their template, but drew its monomials at random subject to those constraints, and drew its coefficient magnitudes (with random signs) directly from the pool of magnitudes seen in the original polynomial (rational field $\mathbb{Q}$). Repeating this 100 times produced our PHC-shaped dataset. The procedure deliberately matches per-polynomial shape without preserving any higher-level algebraic structure such as symmetry, equation coupling, or specific coefficient relationships. The motivation is to test whether the difference in heuristic performance between random and applied systems is explained by per-polynomial shape alone, or whether the deeper algebraic structure of applied systems is what drives the difference.

\begin{table}[t]
\centering
\caption{Performance on a PHC-shaped synthetic corpus of 100 systems generated by matching each real PHCpack system's per-polynomial shape (monomial count, maximum degree, variable usage, coefficient magnitudes). Run-time Mean, median, and strict wins are computed on the 57 systems all three strategies solved. \textbf{\textcolor{blue}{Blue bold text}} marks the best value in each column.}
\label{tab:phcsynthetic}
\begin{tabular}{l c c c c c}
\toprule
Strategy & Finished & Run-time & Mean & Median & Strict wins \\
\midrule
Degree   & 66/100                        & {\color{blue}\textbf{10\,m 31\,s}} & {\color{blue}\textbf{11\,s}} & 0.56\,s                        & 15 \\
Sugar    & {\color{blue}\textbf{68/100}} & 11\,m 4\,s                        & 12\,s                        & 0.48\,s & 5 \\
Normal  & 57/100                        & 28\,m 44\,s                    & 30\,s                        & {\color{blue}\textbf{0.47\,s}}                        & {\color{blue}\textbf{23}} \\
Entropy  & 66/100                        & 11\,m 11\,s                        & 11\,s                        & 0.80\,s                        & 4 \\
\bottomrule
\medskip
{\footnotesize 10 systems were ties, no unique winner}
\end{tabular}

\end{table}

Table \ref{tab:phcsynthetic} shows a somewhat similar behaviour to what we observed for the original PHCpack dataset with Entropy again performing worse than Sugar. However, here all three strategies finish similar numbers of systems (66--68 of 100) except Normal which finishes 57 out of 100 systems. A similar trend in runtime can be observed with Degree, Sugar and Entropy being in the range of 10--12m but Normal taking significantly more time of nearly 29 minutes. The runtimes show no strategy completely dominating, in contrast to the results in Tables~\ref{tab:random-results} and \ref{tab:phc-results}. 

Thus we conclude that while the right choice of selection strategy is highly dependent on the structure of the polynomials as defined in \cite{PSH_20}, the dataset parameters used there do not fully capture the difference.  This highlights the need for as broad as possible range of synthetic datasets to use when evaluating (or training machine learning models for the problem).

\subsection{The Effect of the Coefficient Field}

Beyond the shape of the polynomial systems, the other factor that distinguished the PHCpack dataset experiment in Section \ref{section:PHC_results} from the random datasets experiment reported in Section~\ref{section:Random_results} was the coefficient field: the random systems were defined over $\mathbb{Z}/32003\mathbb{Z}$, while PHCpack systems were defined over $\mathbb{Q}$. To isolate the effect of the coefficient field, we re-ran the PHCpack experiments over $\mathbb{Z}/32003\mathbb{Z}$, in which each coefficient in the input polynomials was first converted to its integer representative in $\{0, \dots, 32002\}$. 

The results are shown in Table \ref{tab:phc-results-ZZ}.  Notably all the strategies end up with fewer timeouts, consistent with modular arithmetic being substantially cheaper than rational arithmetic and thus reducing per-system runtime across the board.  The trends observed in Section \ref{section:PHC_results} is preserved in this case as well. Sugar dominated in all aspects. The only difference we see here is that in Table \ref{tab:phc-results}, Degree ended up with the most number of wins, however in  Table \ref{tab:phc-results-ZZ} Sugar dominates the wins as well. 
For Entropy we observe similar results, however in this case these differences become less pronounced in terms of runtimes compared to Table \ref{tab:phc-results}.

We conclude that the comparative performance of the selection strategies is more driven by polynomial structure rather than by the cost of the underlying coefficient arithmetic.  

\begin{table}
\centering
\caption{Performance on the PHCpack benchmark suite in the integer field (94 real-world polynomial systems drawn from engineering and scientific applications; 600\,s per-instance
timeout). \#TO is the number of timeouts.
``Runtime incl. TO'' aggregates wall-clock time with timeouts charged at the 600\,s limit; ``Runtime w/o TO'' aggregates only over the 55 instances solved by all three strategies. Mean and median are computed over each strategy's own non-timed-out runs. The best performance of each metric is in \textbf{\textcolor{blue}{blue bold text}}.}

\label{tab:phc-results-ZZ}
\begin{tabular}{l c c r r r c}
\toprule
Strategy & \#TO
  & \shortstack{Runtime\\incl. TO} & \shortstack{Runtime\\w/o TO}
  & \shortstack{Mean\\w/o TO} & \shortstack{Median\\w/o TO} & \shortstack{\# of wins}
\\
\midrule
Normal  & {\color{blue}\textbf{19}}                       
    & 3h 35m 1s                        
    & 19m 17s                        
    &15.634\,s
    & 0.329\,s  
    & 12 \\
Sugar   & {\color{blue}\textbf{19}} 
    & {\color{blue}\textbf{3h 28m 41s}} 
    & {\color{blue}\textbf{16m 16s}}                      
    & {\color{blue}\textbf{13.183\,s}}                        
    & {\color{blue}\textbf{0.269\,s}}
    & {\color{blue}\textbf{44}}\\
Degree  & 20                       
    & 3h 45m 4s                        
    & 25m 4s                        
    &20.328\,s
    & 0.294\,s   
    & 13\\
Entropy & {\color{blue}\textbf{19}}                        
    & 3h 41m 44s                        
    & 23m 39s
    & 19.178\,s                        
    & 0.362\,s 
    & 5 \\
\bottomrule
\end{tabular}
\end{table}

\section{Discussion and Conclusion}

The experiments reported in Sections \ref{section:Random_results} and \ref{section:PHC_results} tell us opposite conclusions: we view this as highly informative rather than contradictory! On the random systems where the polynomials are sparse and more heterogenous, Entropy decisively outperform classical heuristics, Sugar, Normal and Degree. This is consistent throughout the reported metrics (overall runtime, runtime on commonly solved systems, no. of wins, mean and median). On the other hand for the PHCpack dataset, where polynomials are denser with a rough average of 7.5 terms and applied algebraic structure is present, the conclusion on strategies reverses. In this case Sugar achieves the lowest aggregate runtime and wins on most of the systems with Degree at second place.

The new Entropy strategy ranks S-pairs by the Shannon entropy of the term-degree distribution of the resulting S-polynomial which is a single scalar summarising how spread or concentrated the degrees of its terms are. The information gained from this heuristic is dependent on whether the produced S-polynomials meaningfully differ by more than a multiplicative constant. On the random distributions reported in Section \ref{section:Random_results}, candidate S-pairs vary in this respect and the entropy ranking discriminates among them effectively. The various measures of leading monomial degree used by the classical heuristics have less leverage in these distributions: Entropy can read a degree-distribution difference that Sugar does not see. 

On the PHCpack dataset the picture inverts. In a small number of PHCpack systems where Entropy performs substantially worse (red points in Figure 1), we inspected the S-polynomials produced and observed that they tend to have similar degree distributions. In each case, the term-degree entropy of many candidate S-polynomials is close to the same value, and the heuristic cannot discriminate among the corresponding pairs, thus falling back to selecting whichever pair comes first in the queue, landing on pairs whose reductions happen to be expensive. Sugar does not depend on the degree distribution of the S-polynomial and survives this, but it does benefit from the additional qualitative information of Normal as a tie-break.

Our experiments on the synthetic PHCpack-shaped distribution in Section \ref{sec:PHCpackShaped}, in which random polynomials are generated to match PHCpack's polynomial structure properties but without the higher-order algebraic structure, brings the three strategies to within roughly 5\% of each other on aggregate run-time. This suggests that while a good deal of the difference between the results of Sections \ref{section:Random_results} and \ref{section:PHC_results} can be explained by the very different structures of polynomials, the language used to describe this, that we follow from \cite{PSH_20}, needs more nuance.

We characterise Entropy's strength as inputs where term-degree distributions of S-polynomials vary meaningfully between candidate pairs, and its weakness as inputs where polynomial size or algebraic structure produces S-polynomials with near-identical degree distributions. 
We do not claim Entropy to be a state-of-the-art selection strategy. Our contribution is rather to present evidence that information-theoretic measures of S-polynomial structure carry useful selection signal, and to give clear evidence that different selection strategies perform radically differently on different shaped datasets, necessitating further research.

The Sugar and Entropy signals clearly draw different kinds of information from the polynomials: each is informative in distributions where the other is not. Hence a natural extension would be to combine the information into one strategy with a weighted combination or a two-stage selector.  Or indeed, a machine learning model may be able to process these, and many other signals, to give an optimal selection on a wide range of problems.

\begin{credits}

\subsubsection{Data Access Statement.}  Code and datasets for experiments are public at: \\
{\scriptsize \url{https://doi.org/10.5281/zenodo.20556193}}.

\subsubsection{\ackname} US is supported by a joint scholarship from Coventry University and Deakin University.

\subsubsection{\discintname}
The authors have no competing interests to declare that are relevant to the content of this article.
\end{credits}

\bibliographystyle{splncs04}
\bibliography{references}

@InProceedings{Buchberger_1979,
    author    = {Buchberger, Bruno},
    title     = {A criterion for detecting unnecessary reductions in the construction of {G}r{\"o}bner-bases},
    booktitle = {Symbolic and Algebraic Computation},
    publisher = {Springer Berlin Heidelberg},
    address   = {Berlin, Heidelberg},
    pages     = {3-21},
    isbn      = {978-3-540-35128-3},
    year      = {1979},
    doi       = {10.1007/3-540-09519-5_52},
}

@inproceedings{Giovini_et_al_1991,
	author = {Giovini, Alessandro and Mora, Teo and Niesi, Gianfranco and Robbiano, Lorenzo and Traverso, Carlo},
	title = {“{One} sugar cube, please” or selection strategies in the {Buchberger} algorithm},
	booktitle = {Proceedings of the 1991 international symposium on {Symbolic} and algebraic computation},
	publisher = {ACM},
	pages = {49-54},
	year = {1991},
	doi = {10.1145/120694.120701},
	isbn = {9780897914376},
	address = {Bonn West Germany},
}

@article{Gebauer_Moller_1988,
	author = {Gebauer, R\"udiger and M\"oller, H. Michael},
	title = {On an installation of {B}uchberger's algorithm},
	journal = {Journal of Symbolic Computation},
	volume = {6},
	number = {2},
	pages = {275-286},
	year = {1988},
	doi = {10.1016/S0747-7171(88)80048-8},
}

@inproceedings{FJ_2002,
	address = {Lille France},
	title = {A new efficient algorithm for computing {Gröbner} bases without reduction to zero (\textit{{F}}$_{\textrm{5}}$)},
	isbn = {9781581134841},
	doi = {10.1145/780506.780516},
	language = {en},
	urldate = {2026-01-02},
	booktitle = {Proceedings of the 2002 international symposium on {Symbolic} and algebraic computation},
	publisher = {ACM},
	author = {Faugère, Jean Charles},
	month = jul,
	year = {2002},
	pages = {75--83},
}

@article{MEMA_1982,
	title = {The complexity of the word problems for commutative semigroups and polynomial ideals},
	volume = {46},
	copyright = {https://www.elsevier.com/tdm/userlicense/1.0/},
	issn = {00018708},
	doi = {10.1016/0001-8708(82)90048-2},
	language = {en},
	number = {3},
	urldate = {2026-01-02},
	journal = {Advances in Mathematics},
	author = {Mayr, Ernst W. and Meyer, Albert R.},
	month = dec,
	year = {1982},
	pages = {305--329},
}

@mastersthesis{MC_2004,
  author       = {McKay, Clinton E.},
  title        = {An Analysis of Improvements to Buchberger's Algorithm for Gr{\"o}bner Basis Computation},
  school       = {University of Maryland, College Park},
  year         = {2004},
  month        = dec,
  address      = {College Park, MD},
  type         = {Master's thesis},
  url          = {https://drum.lib.umd.edu/items/07444dba-4078-430d-b77f-fe8fb25ba3bf},
}

@inproceedings{KHI_23,
	author = {Kera, Hiroshi and Ishihara, Yuki and Kambe, Yuta and Vaccon, Tristan and Yokoyama, Kazuhiro},
    title = {Learning to compute {G}r\"{o}bner bases},
    year = {2024},
    isbn = {9798331314385},
    publisher = {Curran Associates Inc.},
    address = {Red Hook, NY, USA},
    booktitle = {Proceedings of the 38th International Conference on Neural Information Processing Systems},
    articleno = {1044},
    numpages = {47},
    url = {https://openreview.net/forum?id=ZRz7XlxBzQ},
    location = {Vancouver, BC, Canada},
    series = {NIPS '24}
}

@inproceedings{PSH_20,
  title      = {Learning Selection Strategies in {B}uchberger’s Algorithm},
  author     = {Peifer, Dylan and Stillman, Michael and Halpern-Leistner, Daniel},
  booktitle  = {Proceedings of the 37th International Conference on Machine Learning},
  pages      = {7575--7585},
  year       = {2020},
  publisher  = {PMLR},
  address    = {Vienna, Austria},
  url        = {https://proceedings.mlr.press/v119/peifer20a.html},
  urldate    = {2025-01-04},
  eventtitle = {International Conference on Machine Learning},
  langid     = {english},
}

@book{Cox_et_al_2015,
    author = {Cox, David A. and Little, John and O'Shea, Donal},
    title = {Ideals, Varieties, and Algorithms: An Introduction to Computational Algebraic Geometry and Commutative Algebra},
    publisher = {Springer International Publishing},
    address = {Cham},
    year = {2015},
    isbn = {978-3-319-16721-3},
    doi = {10.1007/978-3-319-16721-3_2},
}

@article{Buch_06,
    title = {Bruno {B}uchberger’s {PhD} thesis 1965: An algorithm for finding the basis elements of the residue class ring of a zero dimensional polynomial ideal},
    journal = {Journal of Symbolic Computation},
    volume = {41},
    number = {3},
    pages = {475-511},
    year = {2006},
    issn = {0747-7171},
    doi = {10.1016/j.jsc.2005.09.007},
    author = {Bruno Buchberger},
}

@article{VPHCpack,
  author  = {Verschelde, Jan},
  title   = {Algorithm 795: {PHC}pack: A General-Purpose Solver for Polynomial Systems by Homotopy Continuation},
  journal = {ACM Transactions on Mathematical Software},
  volume  = {25},
  number  = {2},
  pages   = {251--276},
  year    = {1999},
  doi     = {10.1145/317275.317286},
}

@inproceedings{CMF_08,
    author = {Caboara, Massimo and Caruso, Fabrizio and Traverso, Carlo},
    title = {Gr\"{o}bner bases for public key cryptography},
    year = {2008},
    isbn = {9781595939043},
    publisher = {Association for Computing Machinery},
    address = {New York, NY, USA},
    doi = {10.1145/1390768.1390811},
    booktitle = {Proceedings of the Twenty-First International Symposium on Symbolic and Algebraic Computation},
    pages = {315–324},
    numpages = {10},
    location = {Linz/Hagenberg, Austria},
    series = {ISSAC '08}
}

@book{SMM_09,
  title     = {Gr{\"o}bner Bases, Coding, and Cryptography},
  editor    = {Sala, Massimiliano and Mora, Teo and Perret, Ludovic and Sakata, Shojiro and Traverso, Carlo},
  publisher = {Springer},
  address   = {Berlin, Heidelberg},
  year      = {2009},
  isbn      = {978-3-540-93805-7},
  doi       = {10.1007/978-3-540-93806-4},
}

@article{DK_86,
    title = {Using {G}röbner bases to reason about geometry problems},
    journal = {Journal of Symbolic Computation},
    volume = {2},
    number = {4},
    pages = {399-408},
    year = {1986},
    issn = {0747-7171},
    doi = {10.1016/S0747-7171(86)80007-4},
    author = {Deepak Kapur},
}

@Article{PdREC24,
    author = {Lynn Pickering and Tereso {Del Rio Almajano} and Matthew England and Kelly Cohen},
    title = {Explainable {AI} Insights for Symbolic Computation: {A} case study on selecting the variable ordering for cylindrical algebraic decomposition},
    journal = {Journal of Symbolic Computation},
    volume = {123},
    pages = {102276},
    doi = {10.1016/j.jsc.2023.102276},
    year = {2024},
}

@Article{dRE24,
    author = {Tereso {del R{\i}o} and Matthew England},
    title = {Lessons on Datasets and Paradigms in Machine Learning for Symbolic Computation: {A} Case Study on {CAD}},
    journal = {Mathematics in Computer Science},
    volume = {18},
    article_number = {17},
    pages = {1--27},
    doi = {10.1007/s11786-024-00591-0},
    year = {2024},
}

@InProceedings{BEG24b,
    author = {Rashid~Barket and Matthew~England and J\"{u}ergen~Gerhard},
    title = {Symbolic Integration Algorithm Selection with Machine Learning: {LSTM}s vs Tree {LSTM}s},
    editor = {Buzzard, Kevin and Dickenstein, Alicia and Eick, Bettina and Leykin, Anton and Ren, Yue},
    booktitle = {Mathematical Software -- ICMS 2024},
    series = {Lecture Notes in Computer Science},
    volume = {14749},
    pages = {167--175},
    publisher = {Springer Nature Switzerland},
    address={Cham},
    doi = {10.1007/978-3-031-64529-7_18},
    year = {2024}
}

@InProceedings{BSEG24,
    author = {Barket, Rashid and Shafiq, Uzma and England, Matthew and Gerhard, J\"{u}rgen},
    title = {Transformers to Predict the Applicability of Symbolic Integration Routines},
    booktitle = {The 4th Workshop on Mathematical Reasoning and {AI} ({MATH-AI}) at {NeurIPS}'24}, 
    url = {https://openreview.net/forum?id=b2Ni828As7},
    year = {2024},
    publisher = {NeurIPS 2024},
    address   = {Vancouver, Canada},
    numpages = {9},
}

@article{BB_01,
  author       = {Buchberger, Bruno},
  title        = {{Gr{\"o}bner Bases and Systems Theory}},
  journal      = {Multidimensional Systems and Signal Processing},
  volume       = {12},
  number       = {3-4},
  pages        = {223--251},
  year         = {2001},
  doi          = {10.1023/A:1011949421611},
  publisher    = {Springer Science+Business Media},
}

@PhdThesis{Peifer2021,
     author = {Dylan James Peifer},
     title = {Reinforcement Learning in {B}uchberger's Algorithm},
     school = {Cornell University},
     url = {http://doi.org/10.7298/4fpv-bn09},
     year = {2021},
}

@article{Shannon_1948,
    author = {Shannon, C. E.},
    title = {A mathematical theory of communication},
    journal = {The Bell System Technical Journal},
    volume = {27},
    number = {3},
    pages = {379-423},
    year = {1948},
    doi = {10.1002/j.1538-7305.1948.tb01338.x}
}

@article{Sadeghimanesh_Feliu_2019,
    author = {Sadeghimanesh, AmirHosein and Feliu, Elisenda},
    title = {Gr\"obner bases of reaction networks with intermediate species},
    journal = {Advances in Applied Mathematics},
    volume = {107},
    number = {},
    pages = {74-101},
    year = {2019},
    doi = {https://doi.org/10.1016/j.aam.2019.02.006},
}

@Article{SymPy2017,
 author = {A.~Meurer and C.~P.~Smith and M.~Paprocki and O.~\v{C}ert\'{i}k and S.~B.~Kirpichev and M.~Rocklin and A.~Kumar and S.~Ivanov and J.~K.~Moore and S.~Singh and T.~Rathnayake and S.~Vig and B.~E.~Granger and R.~P.~Muller and F.~Bonazzi and H.~Gupta and S.~Vats and F.~Johansson and F.~Pedregosa and M.~J.~Curry and A.~R.~Terrel and \v{S}.~Rou\v{c}ka and A.~Saboo and I.~Fernando and S.~Kulal and R.~Cimrman and A.~Scopatz},
 title = {{SymPy}: {S}ymbolic computing in {P}ython},
 volume = {3},
 pages = {e103},
 journal = {PeerJ Computer Science},
 url = {https://doi.org/10.7717/peerj-cs.103},
 year = {2017}
}

@InCollection{HEWDPB14,
 author = {Z.~Huang and M.~England and D.~Wilson and J.~H.~Davenport and L.~Paulson and J.~Bridge},
 title = {Applying machine learning to the problem of choosing a heuristic to select the variable ordering for cylindrical algebraic decomposition},
 booktitle = {Intelligent Computer Mathematics},
 editor = {S.~M.~Watt and J.~H.~Davenport and A.~P.~Sexton and P.~Sojka and J.~Urban},
 series = {Lecture Notes in Artificial Intelligence},
 volume = {8543},
 pages = {92--107},
 publisher = {Springer International},
 url = {http://dx.doi.org/10.1007/978-3-319-08434-3_8},
 year = {2014}
}

@InProceedings{SNB23,
 author = {Vaibhav Sharma and Abhinav Nagpal and {Muhammed Fatih} Balin},
 title = {{SIRD}: {S}ymbolic Integration Rules Dataset},
 booktitle = {Proceedings of the 3rd Workshop on Mathematical Reasoning and AI (MATH-AI 2023) at NIPS 2023},
 url = {https://mathai2023.github.io/papers/39.pdf},
 year = {2023}
}

@InProceedings{JDLHMZ23,
 author = {Fuqi Jia and Yuhang Dong and Minghao Liu and Pei Huang and Feifei Ma and Jian Zhang},
 title = {Suggesting Variable Order for Cylindrical Algebraic Decomposition via Reinforcement Learning},
 booktitle = {Thirty-seventh Conference on Neural Information Processing Systems (NIPS 2023)},
 url = {https://openreview.net/forum?id=vNsdFwjPtL},
 year = {2023}
}

@Article{HMS25,
 author = {Amir Hashemi and Mahshid Mirhashemi and Werner M. Seiler},
 title = {Machine learning parameter systems, {N}oether normalisations and quasi-stable positions},
 journal = {Journal of Symbolic Computation},
 volume = {126},
 pages = {102345},
 url = {https://doi.org/10.1016/j.jsc.2024.102345},
 year = {2025}
}

\end{document}